\documentclass[11pt]{article} 

\usepackage[utf8]{inputenc} 

\usepackage{hyperref}
\usepackage{cite}
\hypersetup{colorlinks=false}

\usepackage{geometry,amsmath} 
\usepackage{graphicx} 

\numberwithin{equation}{section}

\usepackage{booktabs} 
\usepackage{array} 
\usepackage{paralist} 
\usepackage{verbatim} 
\usepackage{subfig} 

\usepackage{float}

\usepackage{todonotes}

\usepackage{amssymb}
\usepackage{amsmath}
\usepackage{amsthm}
\usepackage{mathrsfs}
\usepackage{bm}

\usepackage{xcolor,pict2e}

\theoremstyle{plain}
\newtheorem{thm}{Theorem}[section]
\newtheorem{cor}[thm]{Corollary}

\newtheorem{prop}[thm]{Proposition}
\newtheorem{defn}[thm]{Definition}

\theoremstyle{definition}
\newtheorem{exm}[thm]{Example}

\newcommand{\es}{\operatorname{S}}

\newcommand{\esK}{\es_\mathbf{K}}
\newcommand{\te}{\operatorname{T}}
\newcommand{\teK}{\te_\mathbf{K}}
\newcommand{\teL}{\te_\mathbf{L}}
\newcommand{\id}{\operatorname{id}}
\newcommand{\pr}{\operatorname{pr}}

\newcommand{\p}{\partial}

\newcommand{\diffi}{\mathrm{D}_{n^i}}
\newcommand{\diffj}{\mathrm{D}_{n^j}}
\newcommand{\diffn}{\mathrm{D}_{n}}
\newcommand{\dbj}{\overline{\mathrm{D}}_{j}}
\newcommand{\dde}{D$\Delta$E\ }
\newcommand{\ddes}{D$\Delta$Es\ }
\newcommand{\ddee}{D$\Delta$E}
\newcommand{\ddese}{D$\Delta$Es}
\newcommand{\mbf}{\mathbf{f}}
\newcommand{\mbn}{\mathbf{n}}
\newcommand{\mbr}{\mathbf{r}}
\newcommand{\mbu}{\mathbf{u}}
\newcommand{\mbv}{\mathbf{v}}
\newcommand{\mbx}{\mathbf{x}}
\newcommand{\mbE}{\mathbf{E}}

\newcommand{\mbJ}{\mathbf{J}}
\newcommand{\mbK}{\mathbf{K}}
\newcommand{\mbL}{\mathbf{L}}
\newcommand{\mbQ}{\mathbf{Q}}
\newcommand{\mbzero}{\mathbf{0}}
\newcommand{\mbone}{\mathbf{1}}
\newcommand{\mblam}{\bm{\lambda}}
\newcommand{\mcA}{\mathcal{A}}
\newcommand{\mcC}{\mathcal{C}}
\newcommand{\mcD}{\mathcal{D}}

\newcommand{\mcL}{\mathcal{L}}
\newcommand{\mcO}{\mathcal{O}}

\newcommand{\mcT}{\mathcal{T}}
\newcommand{\msD}{\mathscr{D}}

\newcommand{\gtt}{\mathtt{g}}
\newcommand{\gttb}{\mathtt{g}}
\newcommand{\ua}{u^\alpha_{\mathbf{0};\mathbf{0}}}
\newcommand{\uaJ}{u^\alpha_{\mathbf{J};\mathbf{0}}}
\newcommand{\uaK}{u^\alpha_{\mathbf{0};\mathbf{K}}}
\newcommand{\uaJK}{u^\alpha_{\mathbf{J};\mathbf{K}}}

\usepackage{fancyhdr} 
\usepackage{sectsty}
\allsectionsfont{\sffamily\mdseries\upshape} 

\title{Transformations, symmetries and Noether theorems for differential-difference equations}
\author{Linyu Peng{$^{1}$\footnote{Adjunct faculty member at the School of Mathematics and Statistics, Beijing Institute of Technology, Beijing 100081, China, and adjunct researcher at the Waseda Institute for Advanced Study, Waseda University, Tokyo 169-8050, Japan. Email: L.Peng@mech.keio.ac.jp } ~ and Peter E. Hydon{$^2$}\footnote{Email: P.E.Hydon@kent.ac.uk} }\vspace{0.4cm}
\\
{\it 1. Department of Mechanical Engineering, Keio University,}\\
{\it Yokohama 223-8522, Japan}\\
{\it 2. School of Mathematics, Statistics and Actuarial Science,}\\
{\it University of Kent, Canterbury CT2 7FS, UK}\\ }

\begin{document}
\maketitle

\abstract{The first part of this paper develops a geometric setting for differential-difference equations that resolves an open question about the extent to which continuous symmetries can depend on discrete independent variables.
For general mappings, differentiation and differencing fail to commute. We prove that there is no such failure for structure-preserving mappings, and identify a class of equations that allow greater freedom than is typical.
	
For variational symmetries, the above results lead to a simple proof of the differential-difference version of Noether's Theorem. We state and prove the differential-difference version of Noether's Second Theorem, together with a Noether-type theorem that spans the gap between the analogues of Noether's two theorems. These results are applied to various equations from physics.

}
\vspace{0.5cm}

\section{Introduction}

In 1918, Emmy Noether's celebrated paper {\it `Invariante Variationsprobleme'} \cite{Noether1918} introduced generalized symmetries of variational problems and established two fundamental theorems. The first of these (known as Noether's Theorem) explains the connection between finite-dimensional Lie algebras of variational symmetry generators and conservation laws of the underlying Euler--Lagrange equations. In classical mechanics, for instance, Noether's Theorem links invariance under translations in time to conservation of energy, while invariance under rotations corresponds to conservation of angular momentum. If the Euler--Lagrange equations are of Kovalevskaya form, there is a bijection between (equivalence classes of) variational symmetries and conservation laws \cite{Olver1993}.

Noether's Second Theorem applies to the other extreme: the variational problem has gauge symmetries depending on arbitrary functions of all independent variables (and their derivatives) if and only if there exist differential relations between the Euler--Lagrange equations.
For a comprehensive history of Noether's two theorems and their generalizations  (up to the end of the twentieth century), see Kosmann-Schwarzbach \cite{Kosmann2011}.

In early 1980s, Maeda \cite{Maeda1981} extended Noether's Theorem to a simple class of ordinary difference equations. Kupershmidt \cite{Kuper1985} introduced the formal approach to difference variational principles, including the inverse problem of the calculus of variations. Noether's Theorem for finite difference equations on a computational mesh was proposed by Dorodnitsyn (e.g. \cite{Dorod2001,Dorod2010}); unlike the continuous case, variational symmetries are not necessarily symmetries of the underlying meshed difference Euler--Lagrange equations. Building on Kupershmidt's work, Hydon \& Mansfield \cite{Hydon2004} constructed the difference variational complex, leading to a general form of Noether's Theorem. For methods of constructing Noether's conservation laws for difference equations and finite element methods, see \cite{Hydon2014,Mansfield2006,Mansfield2017,MRHP2019}. Recently, Hydon \& Mansfield \cite{Hydon2011} derived Noether's Second Theorem for difference equations and bridged the gap between Noether's theorems (for differential and difference equations) by treating variational symmetries that depend on constrained functions.

Differential-difference equations (\ddes) can be used to model mechanical and other systems. They also arise as semi-discretizations of partial differential equations (PDEs) and conversely, as (partial) continuum limits of partial difference equations. There have been various approaches to adapting symmetry methods to \ddes (see \cite{Quispel1992,Levi2006,Levi1991,Yamilov2006,Peng2019}). In 2010, Levi \textit{et al.} \cite{Levi2010} used a limiting argument to show that the obvious adaptation of Lie point symmetry methods to \ddes needs modification when the transformation of a continuous independent variable, $x\in\mathbb{R}$, depends on a discrete independent variable, $n\in\mathbb{Z}$. However, further difficulties arise in this case, as the transformed difference and differential operators do not commute \cite{Peng2017}. This has raised a key question: are there any circumstances in which the transformation of $x$ can depend on $n$?
 
The current paper answers this question by establishing the geometric conditions for a mapping to be a well-defined point transformation of a \ddee. These lead immediately to constraints on Lie point symmetries (see Section \ref{sec:pf}) that were known previously only in special cases. We identify an exceptional class of \ddes whose symmetries appear to violate these constraints, and show how this phenomenon is consistent with the geometric framework for generic \ddese.
 
The second half of the paper addresses \ddes with variational symmetries. Section \ref{sec:NFT} deals with Noether's Theorem. Section \ref{sec:NST} proves Noether's Second Theorem for \ddes whose variational symmetries depend on arbitrary functions of all independent variables. A generalization to variational symmetries that depend on constrained functions is given in Section \ref{sec:NIT}.

\section{Transformations and symmetries of \ddese}
\label{sec:pf}

A solution of a differential or difference equation can be written (locally) as the graph of a function. The distinction between independent and dependent variables leads to a geometric structure (based on prolongation) that determines the class of transformations, a \textit{transformation} being a structure-preserving bijection whose inverse is also structure-preserving. A \textit{symmetry} of a given equation is a transformation that preserves the set of solutions of the equation. 

For \ddese, the corresponding structure is obtained by combining the differential and difference structures. For simplicity, we restrict attention to \ddes that are defined on $\mathbb{R}^{p}\times\mathbb{Z}^{m}$, with continuous independent variables $\mbx=(x^1,\dots,x^{p})$, discrete independent variables $\mbn=(n^1,\dots,n^{m})$, and dependent variables $\mbu=(u^1,\dots,u^q)\in\mathbb{R}^q$. (Domains other than $\mathbb{R}^{p}\times\mathbb{Z}^{m}$ can be dealt with by treating these variables as local coordinates.) The space $\mcT=\mathbb{R}^{p}\times\mathbb{Z}^{m}\times\mathbb{R}^q$ of all independent and dependent variables is called the \textit{total space}.

From here on, all functions are assumed to be locally smooth in each of their continuous arguments, for every $\mbn$. (This avoids the need to discuss technicalities associated with singularities and other discontinuities.) The Einstein summation convention is used throughout.

\subsection{Differential structure} 

For any $\mbn\in\mathbb{Z}^{m}$, the \textit{slice} $\mcT_\mbn=\mathbb{R}^{p}\times\{\mbn\}\times\mathbb{R}^q$ is a continuous space. Every function $\mbu=\mathbf{f}(\mbx,\mbn)$, restricted to this slice, can be prolonged by differentiation as many times as is needed. This gives rise to the infinite jet space, $J^{\infty}(\mcT_\mbn)$, whose `vertical' coordinates $\uaJ$ represent the values of the derivatives of the dependent variables. (The index after the semicolon is used later to indicate values of jet space variables on different slices, the slice at a given $\mbn$ being denoted by $\mbzero$.) Each component $j^i$ of the index $\mbJ=(j^1,\dots,j^{p})$ denotes the number of derivatives with respect to $x^i$. In particular, $\ua=u^{\alpha}$ and the first derivatives of $u^\alpha = f^\alpha(\mbx,\mbn)$ are represented by the coordinate values
\[
u^\alpha_{\mbone_i;\mbzero}=\frac{\partial f^\alpha(\mbx,\mbn)}{\partial x^i}\,,\qquad i=1,\dots,p,
\]
where $\mbone_i$ has $1$ in the $i$-th entry and zeros elsewhere. More generally, the action of the first derivative with respect to $x^i$ on any differentiable function defined on $J^{\infty}(\mcT_\mbn)$ is given by the operator
\[
D_i\big|_{J^{\infty}(\mcT_\mbn)}:=\frac{\partial}{\partial x^i}+u^{\alpha}_{\mbJ+\mbone_i;\mbzero}\,\frac{\partial}{\partial \uaJ}\,.
\]
(For further details on jet space, see \cite{Olver1995,KraVin1999}.)

So far, we have considered only the jet space over a single (arbitrary) slice. This is sufficient: a copy of the same jet space is generated from the slice over every $\mbn$, as $\mbx$ and $\mbn$ are mutually independent. Together, these constitute the \textit{total jet space} $J^{\infty}(\mcT)\cong \mathbb{Z}^{m}\times J^{\infty}(\mcT_\mbn)$.

\subsection{Difference structure}
 
The difference structure arises from the fact that the total space $\mcT$ is preserved by all translations
\[
 \teK:\mcT\rightarrow\mcT,\qquad\teK:(\mbx,\mbn,\mbu)\mapsto(\mbx,\mbn+\mbK,\mbu).
\]
Note that $\teL\circ\teK=\te_{\mbK+\mbL}$ for all $\mbK,\mbL\in\mathbb{Z}^{m}$.
 
The total space is disconnected, but has a representation as a connected space over each $\mbn$, as follows. Each slice is prolonged to include the values of the coordinates on all other slices as coordinates in a Cartesian product, using the pullback of each $u^\alpha$ with respect to each $\teK\in\mathbb{Z}^{m}$. The resulting prolongation space over $\mbn$ is denoted $P(\mcT_\mbn)$; it has vertical coordinates $\uaK$, where
 \[
 \uaK=\teK^*\left(\ua\big|_{\mbn+\mbK}\right).
 \]

It is straightforward to combine the differential and difference structures, as follows. The action of $\teK$ extends immediately to the total jet space:
\[
\teK:J^{\infty}(\mcT)\rightarrow J^{\infty}(\mcT),\qquad\teK:(\mbx,\mbn,\dots,\uaJ,\dots)\mapsto(\mbx,\mbn+\mbK,\dots,\uaJ,\dots).
\]
Similarly to the difference case, one can prolong the jet space over each $\mbn$ by pulling back the jet space coordinates over all $\mbn+\mbK$. This gives the space $P(J^{\infty}(\mcT_\mbn))$ whose vertical coordinates are
\[
\uaJK=\teK^*\left(\uaJ\big|_{\mbn+\mbK}\right);
\]
this is the connected component on $\mbn$ of the \textit{total prolongation space} $P(J^{\infty}(\mcT))\cong\mathbb{Z}^{m}\times P(J^{\infty}(\mcT_\mbn))$.

The composition rule for translations gives the identities
\[
 u^\alpha_{\mbJ;\mbK+\mbL}=\teK^*\left(u^\alpha_{\mbJ;\mbL}\big|_{\mbn+\mbK}\right).
\]
More generally, let $f$ be a (locally smooth) function on $P(J^{\infty}(\mcT))$ and denote its restriction to $P(J^{\infty}(\mcT_\mbn))$ by 
\[
f_\mbn(\mbx,\dots,u^\alpha_{\mbJ;\mbL},\dots):=f(\mbx,\mbn,\dots,u^\alpha_{\mbJ;\mbL},\dots).
\]
The pullback to $P(J^{\infty}(\mcT_\mbn))$ of $f_{\mbn+\mbK}(\mbx,\dots,u^\alpha_{\mbJ;\mbL},\dots)$ is the function
\[
\te_{\mbK}^*f_{\mbn+\mbK}=f(\mbx,\mbn+\mbK,\dots,u^{\alpha}_{\mbJ;\mbK+\mbL},\dots).
\]
Therefore, the action of the translation $\te_{\mbK}$ on the space of locally smooth functions is represented on $P(J^{\infty}(\mcT_\mbn))$ by the \textit{shift operator} $\es_{\mbK}$, defined by $\es_{\mbK}\!f_\mbn=\te_{\mbK}^*f_{\mbn+\mbK}$, which gives
 \begin{equation}
 \es_{\mbK}:f(\mbx,\mbn,\dots,u^\alpha_{\mbJ;\mbL},\dots)\mapsto f(\mbx,\mbn+\mbK,\dots,u^\alpha_{\mbJ;\mbK+\mbL},\dots).
 \end{equation}
Consequently, the derivative with respect to $x^i$ on $J^{\infty}(\mcT)$ is represented on $P(J^{\infty}(\mcT_\mbn))$ by the \textit{total derivative}
\begin{equation}
D_i=\frac{\partial}{\partial x^i}+u^{\alpha}_{\mbJ+\mbone_i;\mbK}\,\frac{\partial}{\partial \uaJK}.
\end{equation}
Crucially, all total derivatives and shift operators commute:
\begin{equation}
D_iD_j=D_jD_i,\qquad D_i\es_\mbK=\es_\mbK D_i,\qquad \es_\mbK\es_\mbL=\es_\mbL\es_\mbK.
\end{equation}
It is convenient to use the following shorthand notation for products of total derivatives:
\[
D_\mbJ=D_1^{j^1}\cdots D_p^{j^p}, \qquad \text{where}\ \mbJ=(j^1,\dots,j^p).
\]

Difference operators on the continuous space $P(J^{\infty}(\mcT_\mbn))$ arise from the ordering of each $n^i$. For any index $\mbK=(k^1,\dots,k^{m})$, the corresponding shift operator is
$\es_{\mbK}=\es_{1}^{k_1}\cdots \es_{m}^{k_{m}}$, where $\es_i:=\es_{\mathbf{1}_i}$ denotes the forward shift with respect to $n^i$. Then the forward difference in the $n^i$-direction is represented on $P(J^{\infty}(\mcT_\mbn))$ by the operator
\[
\diffi:=\es_i-\id,
\]
where $\id$ is the identity mapping. A \textit{difference divergence} is an expression $\mcC$ such that $\mcC=\diffi G^i$; similarly, a \textit{differential-difference divergence} is an expression $\mcC$ of the form
\begin{equation}\label{eq:div}
\mcC =D_iF^i+\diffi G^i.
\end{equation}

The \textit{formal adjoint} of a linear operator $\mcO$ is the unique operator $\mcO^{\dagger}$ such that $f\mcO g-(\mcO^{\dagger}f)g$ is a (differential-difference) divergence for all functions $f$ and $g$ defined on $P(J^{\infty}(\mcT_\mbn))$. In particular,
\[
D_i^{\dagger}=-D_i,\qquad \es_i^{\dagger}=\es_i^{-1},\qquad \id^{\dagger}=\id,\qquad \diffi^{\dagger}=-(\id^{\dagger}-\es_i^{\dagger})=-\es_i^{-1}\diffi\,;
\]
the composition rule $(\mcO_1\mcO_2)^{\dagger}=\mcO_2^{\dagger}\mcO_1^{\dagger}$ determines the adjoint of a product of linear operators. Thus
\[
\es_{\mbK}^{\dagger}=\es_{-\mbK}\,,\qquad D_\mbJ^{\dagger}=(-D)_\mbJ:=(-1)^{j^1+\cdots+j^p}D_\mbJ\,.
\]
The following useful result was proved in \cite{Peng2017}.
\begin{thm}\label{thm:ELker}
	A function on $P(J^{\infty}(\mcT_\mbn))$ is a differential-difference  divergence if and only if it belongs to the kernel of each Euler--Lagrange operator,
	\begin{equation}\label{eq:ELop}
		\mbE_{u^\alpha}=(-D)_\mbJ\es_{-\mbK}\frac{\p}{\p \uaJK}\,,\qquad \alpha=1,\dots,q.
	\end{equation}
	\end{thm}

\subsection{Lie point transformations for \ddes}

Having established the prolongation structure that underpins \ddese, we are now in a position to identify the constraints that must be satisfied by transformations of \ddese. The passive viewpoint is used throughout this section, in which a transformation is treated as a change of coordinates.

A \textit{point transformation} is a transformation of the total space,
\[
\Gamma:\mcT\rightarrow\mcT,\qquad\Gamma:(\mbx,\mbn,\mbu)\mapsto (\hat{\mbx},\hat{\mbn},\hat{\mbu}).
\]
Preservation of $\mathbb{Z}^{m}$ requires a \textit{lattice transformation}, $\hat{\mbn}=A\mbn + \mbn_0$, where $A\in GL_{m}(\mathbb{Z})$ and $\mbn_0$ is constant. (See \cite{Hydon2014} for an overview of lattice transformations of $\mathbb{Z}^m$.) For simplicity, we restrict attention to Lie groups of point transformations. Every one-parameter Lie group of mappings from $\mcT$ to itself can be expressed in the form
\begin{equation}\label{eq:totmap}
\hat{x}^i=x^i+\varepsilon\xi^i(\mbx,\mbn,\mbu)+O(\varepsilon^2),\qquad
\hat{n}^i=n^i,\qquad
\hat{u}^\alpha=u^\alpha+\varepsilon\phi^\alpha(\mbx,\mbn,\mbu)+O(\varepsilon^2).
\end{equation}
As such mappings change only $\mbx$ and $\mbu$, they are represented on $P(J^{\infty}(\mcT_\mbn))$ by using the same shift operators $\esK$ in both the original and transformed coordinates. 

\begin{prop}\label{prop:trans}
	A one-parameter Lie group of mappings \eqref{eq:totmap} prolongs to a transformation group for $P(J^{\infty}(\mcT_\mbn))$ if and only if each $\xi^i$ is independent of $\mbn$ and $\mbu$.
\end{prop}

\begin{proof}
For each $\mbK\in\mathbb{Z}^{m}\backslash\{\mbzero\}$, the mapping \eqref{eq:totmap} yields
\begin{align*}
\esK x^i - x^i =&\left\{\esK\hat{x}^i-\varepsilon\xi^i(\esK\hat{\mbx},\mbn+\mbK,\hat{\mbu}_{\mbzero;\mbK})\right\}-\left\{\hat{x}^i- \varepsilon\xi^i(\hat{\mbx},\mbn,\hat{\mbu}_{\mbzero;\mbzero})\right\}+O(\varepsilon^2)\\
=& \esK\hat{x}^i-\hat{x}^i+\varepsilon\left\{\xi^i(\hat{\mbx},\mbn,\hat{\mbu}_{\mbzero;\mbzero})-\xi^i(\hat{\mbx},\mbn+\mbK,\hat{\mbu}_{\mbzero;\mbK})\right\}+O(\varepsilon^2).
\end{align*}
For the mapping to preserve the structure of $P(J^{\infty}(\mcT_\mbn))$, it must satisfy $\esK\hat{x}^i=\hat{x}^i$ for all $\mbK$, just as $\esK x^i =x^i$. This gives the condition
\begin{equation}\label{eq:xii}
\xi^i(\hat{\mbx},\mbn+\mbK,\hat{\mbu}_{\mbzero;\mbK})=\xi^i(\hat{\mbx},\mbn,\hat{\mbu}_{\mbzero;\mbzero}),
\end{equation}
for all $\mbK$. The coordinates $\hat{\mbu}_{\mbzero;\mbK}$ and $\hat{\mbu}_{\mbzero;\mbzero}$ are distinct, so $\xi^i$ does not depend on $\mbu$. Therefore, from \eqref{eq:xii}, $\xi^i$ is also independent of $\mbn$.

Conversely, every mapping of total space,
\begin{equation}\label{eq:transt}
\hat{x}^i=x^i+\varepsilon\xi^i(\mbx)+O(\varepsilon^2),\qquad
\hat{n}^i=n^i,\qquad
\hat{u}^\alpha=u^\alpha+\varepsilon\phi^\alpha(\mbx,\mbn,\mbu)+O(\varepsilon^2),
\end{equation}
that belongs to a one-parameter Lie group has the same transformed total derivatives, denoted $\hat{D}_i$, on each jet space, because $\hat{\mbx}$ depends on $\mbx$ only (by the standard method for constructing each $\hat{x}^i$ from $\xi^i$). Therefore, each $\hat{D}_i$ commutes with every shift operator. So the prolongation of the (invertible) mapping \eqref{eq:transt} preserves the structure of $P(J^{\infty}(\mcT_\mbn))$, and is thus a transformation.
\end{proof}

The prolongation formula for the transformation group determined by \eqref{eq:transt} can be written most concisely in terms of the characteristic,
$\mbQ:=(Q^1,\dots,Q^q)$, where
\[
Q^\alpha=\phi^\alpha(\mbx,\mbn,\mbu)-\xi^i(\mbx)u^\alpha_{\mbone_i;\mbzero}\,.
\] 
The standard differential prolongation formula (see \cite{Olver1993}) applied to the jet space $J^{\infty}(\mcT_\mbn)$ gives
\[
\hat{u}^\alpha_{\mbJ;\mbzero}=\uaJ+\varepsilon\phi^\alpha_{\mbJ;\mbzero}+O(\varepsilon^2),\quad\text{where}\quad \phi^\alpha_{\mbJ;\mbzero}=D_\mbJ Q^\alpha +\xi^i(\mbx)u^\alpha_{\mbJ+\mbone_i;\mbzero}\,.
\]
This prolongs to $P(J^{\infty}(\mcT_\mbn))$ as follows:
\begin{equation}
\hat{u}^\alpha_{\mbJ;\mbK}=\uaJK+\varepsilon\phi^\alpha_{\mbJ;\mbK}+O(\varepsilon^2),\quad\text{where}\quad \phi^\alpha_{\mbJ;\mbK}=\esK\! D_\mbJ Q^\alpha +\xi^i(\mbx)u^\alpha_{\mbJ+\mbone_i;\mbK}\,.
\end{equation}
Then the infinitesimal generator of the one-parameter Lie group of point transformations \eqref{eq:transt}, that is, the vector field
\[
\mbv=\xi^i(\mbx)\frac{\p}{\p x^i}+\phi^\alpha(\mbx,\mbn,\mbu)\frac{\p}{\p u^\alpha}\,,
\]
prolongs to the following infinitesimal generator on $P(J^{\infty}(\mcT_\mbn))$:
\begin{equation}\label{eq:pgen}
\pr\mbv=\xi^i(\mbx)\frac{\p}{\p x^i}+\phi^\alpha_{\mbJ;\mbK}\frac{\p}{\p \uaJK}=\xi^i(\mbx)D_i+(\esK\! D_\mbJ Q^\alpha) \frac{\p}{\p \uaJK}\,.
\end{equation}

The characteristic $\mbQ$ has its origin in the transformation of graphs. On each slice $\mcT_\mbn$, the graph defined by $\mbu=\mbf(\mbn,\mbu)$ transforms as follows:
\begin{align*}
\hat{u}^\alpha=&f^\alpha(\mbx,\mbn)+\varepsilon\phi^\alpha(\mbx,\mbn,\mbf(\mbx,\mbn))+O(\varepsilon^2)\\
&=f^\alpha(\hat{\mbx},\mbn)+\varepsilon\left(\phi^\alpha(\hat{\mbx},\mbn,\mbf(\hat{\mbx},\mbn))-\xi^i(\hat{\mbx})\frac{\p f^\alpha(\hat{\mbx},\mbn)}{\p \hat{x}^i}\right)+O(\varepsilon^2).
\end{align*}
From the active viewpoint (obtained by dropping the carets), the transformation maps $\mbu=\mbf(\mbx,\mbn)$ to $\mbu=\mathbf{h}(\mbx,\mbn)$, where
\begin{equation}\label{eq:transg}
h^\alpha(\mbx,\mbn)=f^\alpha(\mbx,\mbn)+\varepsilon Q^\alpha\big|_{[\mbu=\mbf(\mbx,\mbn)]}+O(\varepsilon^2).
\end{equation}
(Here and henceforth, square brackets around an expression denotes the expression and a finite number of its prolongations, as needed.)
In the same way, every prolongation of $\mbu=\mbf(\mbx,\mbn)$ is mapped to the corresponding prolongation of \eqref{eq:transg}, giving rise to the terms $\esK\! D_\mbJ Q^\alpha$ at first order in $\varepsilon$. This action on graphs leads to the evolutionary representative of the transformation on $P(J^{\infty}(\mcT_\mbn))$, which is the prolongation of the mapping
\[
(\mbx,\mbn,\mbu)\longmapsto\left(\mbx,\mbn,\mbu+\varepsilon\mbQ+O(\varepsilon^2)\right).
\]
The infinitesimal generator for the evolutionary representative is
\begin{equation}\label{eq:Qgen}
\operatorname{pr}\mbv_\mbQ= (\esK\! D_\mbJ Q^\alpha) \frac{\p}{\p \uaJK}\,.
\end{equation}

\begin{figure}[t]
	\centering
	\includegraphics[scale=0.5]{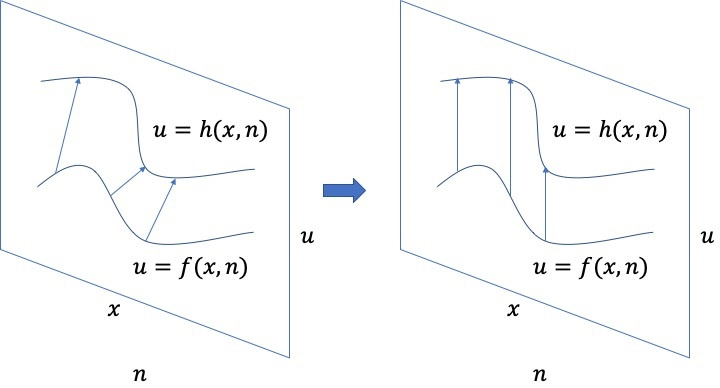}
	\caption{A local Lie point transformation (left) and its evolutionary representative (right).}
	\label{fig:evo}
\end{figure}

Figure 1 illustrates (for $p=m=q=1$) the distinction between the action on a graph of a point transformation that changes $\mbx$, and the corresponding action of its evolutionary representative. A graph on $\mcT$ is \textit{invariant} under the transformation if the restriction of $\mbQ$ to the prolonged graph is zero (on every slice). Equivalently, the graph prolonged to $P(J^{\infty}(\mcT_\mbn))$ is invariant if its restriction to every $\esK\!\mbQ$ is zero. If $\mbQ=0$, all graphs are invariant; the transformation is trivial (on graphs), merely moving points along each graph. From \eqref{eq:pgen}, every one-parameter Lie group of trivial transformations has a generator of the form $\xi^i(\mbx)D_i$. To summarize, $\operatorname{pr}\mbv_\mbQ$ determines the local action of the transformation group on graphs, whereas $\pr \mbv$ determines the action on points. For most purposes, including Noether's theorems, it is enough to know the action on graphs.

Any generator of the form \eqref{eq:Qgen}, where $\mbQ$ depends on $(\mbx,\mbn,[\mbu])$, is called a \textit{generalized transformation}. Noether introduced such transformations for differential equations (see \cite{Noether1918,Olver2018}); they are central to Noether's Theorem.

\subsection{Symmetries of \ddese}
Given a system of \ddese,
\begin{equation}\label{eq:ddesym}
\mcA:=\left(\mcA_{1}(\mbx,\mbn,[\mbu]),\dots,\mcA_\ell(\mbx,\mbn,[\mbu]) \right)=0,
\end{equation}
one can find Lie point and generalized symmetries by solving the {\bf linearized symmetry condition} (LSC),
\begin{equation}\label{LSClong}
\pr \mbv(\mcA_l)=0\ \text{ on all solutions of }  \mcA=0,\qquad l=1,\dots,\ell.
\end{equation}
Here and for the rest of the paper, we use $\mbv$ for all generators, whether or not they are in evolutionary form. Given the restriction to solutions (each of which defines a graph), any trivial component $\xi^i(\mbx)D_i$ contributes nothing to the LSC.

To calculate symmetries, it is necessary to recast the LSC as an equation. For most scalar \ddese, it is obvious how to do this. However, complications occur for some systems. To address these, we adapt the concept of a positive ranking from differential algebra to the differential-difference case.

\begin{defn}
A positive ranking of the variables $\uaJK$ is a total order $\prec$ satisfying (for each i):
\begin{enumerate}
	\item $\uaJK\prec D_i\uaJK\,(\,=u^\alpha_{\mbJ+\mbone_i;\mbK}),$
	\item $\uaJK\prec \es_i\!\uaJK\,(\,=u^\alpha_{\mbJ;\mbK+\mbone_i}),$
	\item $\uaJK\prec u^\beta_{\mathbf{I};\mathbf{L}} \Longrightarrow D_i\uaJK\prec D_iu^\beta_{\mathbf{I};\mathbf{L}}\,,$
	\item $\uaJK\prec u^\beta_{\mathbf{I};\mathbf{L}} \Longrightarrow \es_i\!\uaJK\prec \es_i\! u^\beta_{\mathbf{I};\mathbf{L}}\,.$
\end{enumerate}
\end{defn}

The \textit{leading variable} in $\mcA_l=0$ is its highest-ranked variable, which we denote by $U_{l}$. The equation is solved for that variable if $\mcA_l=U_{l}-\omega_l$ and the leading variable in $\omega_l$ is ranked lower than $U_{l}$. By the defining properties of a positive ranking, every prolongation of $\mcA_l=0$ has the corresponding prolongation of $U_{l}$ as its leading variable.

Suppose that every equation $\mcA_l=0$ is solved for its leading variable with respect to a given positive ranking, and that none of the $\ell$ leading variables coincide with (or are derivatives or shifts of) any other leading variable. Then $\omega=(\omega_1,\dots,\omega_\ell)$ and its prolongations can be substituted for $U=(U_1,\dots,U_\ell)$ and its prolongations in \eqref{LSClong}, turning the LSC into the following system of equations:
\begin{equation}\label{eq:LSC}
	\pr \mbv(\mcA_l)\Big|_{[U=\omega]}=0,\qquad l=1,\dots,\ell.
\end{equation}
For point symmetries, the form of $\mbQ$ is highly restricted, making \eqref{eq:LSC} an overdetermined linear system that can be solved by a combination of differential elimination and methods for linear equations. Generalized symmetries are found by the same approach, after restricting $\mbQ$ as needed.

There is considerable freedom to choose a positive ranking that is appropriate for a given system. Systems that are in generalized Kovalevskaya form (see \cite{Olver1993,Hydon2014,Peng2017}) have a positive ranking based on derivatives or shifts with respect to one independent variable. Systems with gauge symmetries typically do not have a preferred direction, so a positive ranking based on the overall differential and/or difference order should be used.

Symmetries have many uses. Generalized symmetries have been used to classify integrable \ddes (see \cite{Yamilov2006,Gari2017,Gari2018}). Both point and generalized symmetries can be used to derive group-invariant solutions, which are solutions satisfying the additional condition $\mbQ=\mbzero$.

\begin{exm}
	The nonlinear Schr\"{o}dinger (NLS) equation,
	\[
	\operatorname{i}\Psi_{t}+\Psi_{xx}+|\Psi|^2\Psi=0,
	\]
	has a well-known (non-integrable) spatial semi-discretization (with step size $h$) that amounts to the following system of \ddes for $u=\mathrm{Re}\{\Psi\},v=\mathrm{Im}\{\Psi\}$ as functions of $(t,n)$:
	\begin{equation}\label{eq:NLSap}
	\begin{split}
	u'&+h^{-2}\left(v_1-2v+v_{-1}\right)+\left(u^2+v^2\right)v=0,\\
	v'&-h^{-2}\left(u_1-2u+u_{-1}\right)-\left(u^2+v^2\right)u=0;
	\end{split}
	\end{equation}
	here $u'=D_tu$, the $j$-th forward shift of $u$ is denoted $u_j$, and similarly for $v$.
	This system is solved for its leading variables in any positive ranking that ranks derivatives higher than shifts. So the LSC can be solved; its Lie algebra of point symmetry generators is spanned by
	\[
	\mbv_1=\frac{\p}{\p t}\,,\qquad \mbv_2=v\frac{\p}{\p u}-u\frac{\p}{\p v}\,.
	\]
	Let us seek solutions that are invariant under the group generated by $X_1-\gamma X_2$, whose characteristic is
	\[
	\mbQ=(-\gamma v-u',\gamma u-v').
	\]
	The general solution of the invariance condition $\mbQ=\mbzero$ is
	\begin{equation}\label{eq:NLSinv}
	u=g(n)\cos(\gamma t+f(n)),\qquad v=g(n)\sin(\gamma t+f(n)),
	\end{equation}
	where $f$ and $g$ are arbitrary functions. Substituting \eqref{eq:NLSinv} into \eqref{eq:NLSap} gives a system of linear ordinary difference equations for $f$ and $g$. The solution process is messy and has several branches, each of which simplifies to the following family of invariant solutions:
	\[
	u=\left\{\gamma +\left(2h^{-1}\sin\alpha\right)^2\right\}^{1/2}\cos(\gamma t+2\alpha n+\beta),\qquad
	v=\left\{\gamma +\left(2h^{-1}\sin\alpha\right)^2\right\}^{1/2}\sin(\gamma t+2\alpha n+\beta).
	\]
	Here $\beta$ is an arbitrary constant, and $\alpha$ and $\gamma$ are constrained only by the requirement for $u$ and $v$ to be real-valued. When $h$ is sufficiently small, such solutions are a good approximation to the corresponding group-invariant solutions of the NLS equation (which have the same modulus), but for larger $h$, the phase error grows linearly with $n$.
\end{exm}

\subsection{Lie point symmetries of partitioned \ddes}
Proposition \ref{prop:trans} establishes that Lie point transformations preserve
$P(J^{\infty}(\mcT_\mbn))$ if and only if the transformation of the continuous independent variables does not depend on $\mbn$, either explicitly or implicitly through $\mbu$. For most \ddese, $P(J^{\infty}(\mcT_\mbn))$ is the appropriate prolongation space. However, there is an exceptional class of \ddes for which this is not so; these are related to partitioned difference equations (see \cite{Hydon2014}). Symmetries of such \ddes include transformations of $\mbx$ that are periodic in $\mbn$.
\begin{defn}
	A partitioned system of \ddes on $\mathbb{Z}^m$ is a system that can be defined on a periodic sublattice $\mcL_{\mbr}=(r_1\mathbb{Z})\times\cdots\times(r_m\mathbb{Z}),\ r_\mu\in\mathbb{N}$, where at least one $r_\mu$ is not $1$.
\end{defn}
From here on, we assume that the sublattice is aligned with the discrete coordinates $n_i$; this can be achieved by applying a lattice transformation (if necessary) at the outset. Then the $\mu$-th component of any translation between the points on $\mcL_{\mbr}$ is an integer multiple of $r_\mu$.

Without loss of generality, assume that each $r_\mu$ is maximal, so that the system of \ddes on $\mcL_{\mbr}$ is not partitioned. Then $\mathbb{Z}^m$ is covered by $r=r_1r_2\cdots r_m$ copies of $\mcL_{\mbr}$, and solutions of the system on any two copies are entirely unrelated. The prolongation space
$P(J^{\infty}(\mcT_\mbn))$ includes shifts that are not used in the copy of $\mcL_{\mbr}$ that contains $\mbn$, so a symmetry of the \dde does not need to be a transformation of $P(J^{\infty}(\mcT_\mbn))$. However, it must be a transformation of the reduced prolongation space $P_\mbr(J^{\infty}(\mcT_\mbn))$ that is obtained by including only shifts that are the pullback of translations between points in $\mcL_{\mbr}$; this is needed to preserve all prolongations of solutions of the \ddee.

\begin{prop}\label{prop:redtrans}
	The prolongation to $P(J^{\infty}(\mcT_\mbn))$ of the one-parameter Lie group of mappings \eqref{eq:totmap} is a transformation of the reduced prolongation space $P_\mbr(J^{\infty}(\mcT_\mbn))$ if and only if each $\xi^i$ is independent of $\mbu$ and satisfies the periodicity condition
	\begin{equation}\label{eq:percon}
	\xi^i\left(\mbx,\mbn+\sum_{\mu=1}^m k_\mu r_\mu \mbone_\mu\right)=\xi^i(\mbx,\mbn),\quad\emph{for all}\,\ k_\mu\in\mathbb{Z}.
	\end{equation}
\end{prop}

\begin{proof}
	The proof follows the proof of Proposition \ref{prop:trans}, except that $\mbK$ in the condition \eqref{eq:xii} is replaced by
	\[
	\mbK_\mbr=\sum_{\mu=1}^m k_\mu r_\mu \mbone_\mu.
	\]
	(Note that $\xi^i$ can be changed by shifts $\esK$ on $P(J^{\infty}(\mcT_\mbn))$ for which $\mbK$ is not of the form $\mbK_\mbr$.)
	The periodicity condition ensures that, in effect, the restriction of $\xi^i$ to $P_\mbr(J^{\infty}(\mcT_\mbn))$ depends on $\mbx$ only. Consequently, the mappings generated by $\xi^i(\mbx,\mbn)D_i$ act trivially on graphs on $P_\mbr(J^{\infty}(\mcT_\mbn))$.
\end{proof}

\begin{cor}\label{cor:partsym}
	For a given partitioned system of \ddese, $\mcA=0$, every infinitesimal generator $\mbv$ of a one-parameter Lie group of point symmetries is of the form 
	\begin{equation}
	\mbv=\xi^i(\mbx,\mbn)\frac{\p}{\p x^i}+\phi^\alpha(\mbx,\mbn,\mbu)\frac{\p}{\p u^\alpha}\,,
	\end{equation}
	and satisfies the LSC \eqref{LSClong} subject to the periodicity condition \eqref{eq:percon}, with
	\begin{equation}
	\pr \mbv=\xi^i(\mbx,\mbn)\frac{\p}{\p x^i}+\phi^\alpha_{\mbJ;\mbK}\frac{\p}{\p \uaJK},\qquad \phi^\alpha_{\mbJ;\mbK}=\esK\! D_\mbJ\left\{\phi^\alpha(\mbx,\mbn,\mbu)-\xi^i(\mbx,\mbn)u^\alpha_{\mbone_i;\mbzero}\right\}+\xi^i(\mbx,\mbn)u^\alpha_{\mbJ+\mbone_i;\mbK}.
	\end{equation}
The group action on solutions is determined by the evolutionary representative \eqref{eq:Qgen}, where now the characteristic has components
\begin{equation}
Q^\alpha=\phi^\alpha(\mbx,\mbn,\mbu)-\xi^i(\mbx,\mbn)u^\alpha_{\mbone_i;\mbzero}\,.
\end{equation}
\end{cor}

\begin{exm}
Consider the \dde
\begin{equation}\label{eq:partu2}
u'=\frac{u_2}{u},
\end{equation}
where $u'=D_xu$, which is partitioned and has $r=2$. Every Lie point symmetry generator,
\[
\mbv=\xi(x,n)\frac{\p}{\p x}+\phi(x,n,u)\frac{\p}{\p u}\,,
\]
is subject to the periodicity condition $\xi(x,n+2k)=\xi(n,k),\ k\in\mathbb{Z}$. 
The LSC, subject to the periodicity condition, yields the Lie algebra of infinitesimal generators spanned by
\begin{equation}
\begin{aligned}
\mbv_1=\frac{\p}{\p x},&\quad \mbv_2=x\frac{\p}{\p x}+u\frac{\p}{\p u},\quad \mbv_3=2^{\left \lfloor{n/2}\right \rfloor }u\frac{\p}{\p u},\\
\mbv_4=(-1)^n\frac{\p}{\p x},\quad &\mbv_5=(-1)^n\left(x\frac{\p}{\p x}+u\frac{\p}{\p u}\right),
\quad\mbv_6=(-1)^n2^{\left \lfloor{n/2}\right \rfloor }u\frac{\p}{\p u}.
\end{aligned}
\end{equation}
Here $\left \lfloor{n/2} \right \rfloor $ is the greatest integer less than or equal to $n/2$. The prolonged Lie algebra is six-dimensional on $P(J^{\infty}(\mcT_n))$, but does not preserve the prolongation structure. To see this, it is helpful to consider the one-parameter Lie subgroup generated by $\mbv_5$, namely
\[
\Gamma:(x,n,u)\longmapsto(\hat{x},\hat{n},\hat{u})=\left(\exp\{(-1)^n\varepsilon\}x,\,n,\,\exp\{(-1)^n\varepsilon\}u \right),\qquad\varepsilon\in\mathbb{R}.
\]
On $P(J^{\infty}(\mcT_n))$, the prolongation of $\Gamma$ maps $u_1'$ to
\[
\widehat{u_1'}=\te_1^*(\widehat{u'}|_{n+1})=\te_1^*(u'|_{n+1}).
\]
However, this result is incompatible with the differential prolongation
\[
D_{\hat{x}}(\hat u_1)=D_{\hat{x}}(\exp\{(-1)^{n+1}\varepsilon\}u_1)=\frac{D_x(\exp\{(-1)^{n+1}\varepsilon\}u_1)}{D_x(\exp\{(-1)^{n}\varepsilon\}x)}=\exp\{2(-1)^{n+1}\varepsilon\}u_1',
\]
reflecting the fact that, unlike $D_x$, the transformed derivative $D_{\hat{x}}$ does not commute with the shift operator $\es$. A simpler way to see that the structure is not preserved is to note that $(\es-\id)x=0$, but $(\es-\id)\hat{x}\neq 0$, as in the proof of Proposition \ref{prop:trans}.

Indeed, the prolonged Lie algebra is only three-dimensional on the reduced space $P_2(J^{\infty}(\mcT_n))$, whose structure is preserved; it is isomorphic to the Lie algebra of point symmetry generators for the \dde $u'=u_1/u$.
\end{exm}

Although partitioned \ddes seem somewhat unnatural, their difference equation counterparts do occur in practice. For instance, the discrete potential KdV and cross-ratio equations are each embedded in a Toda-type partitioned Euler--Lagrange equation with $r=2$ (see \cite{Bobenko2002,Rasin07a}).

\section{Noether's Theorem}
\label{sec:NFT}
Noether's Theorem connects symmetries of variational problems and conservation laws of the underlying Euler--Lagrange equations. A restricted version of Noether's Theorem for differential-difference equations was proved in \cite{Peng2017}, with examples including the Volterra equation, the Toda lattice and some semi-discretizations of the KdV equation. We now state Noether's Theorem in full generality, using the evolutionary representative; this applies whether or not the system is partitioned (by Corollary \ref{cor:partsym}). Each $\phi^\alpha$ may depend on finitely many shifts and derivatives of $\mbu$, giving generalized symmetries whose characteristic is $\mbQ(\mbx,\mbn,[\mbu])$.

A \textit{conservation law} for a system of \ddese, $\mcA=0$, is a divergence,
\begin{equation}\label{eq:claw}
\mcC=D_iF^i(\mbx,\mbn,[\mbu])+\diffi G^i(\mbx,\mbn,[\mbu]),
\end{equation}
that is zero on all solutions of the system. A conservation law is \textit{trivial} if it can be written in the form \eqref{eq:claw} with all components $F^i$ and $G^i$ being zero when $[\mcA=0]$. Two conservation laws are \textit{equivalent} if their difference is trivial. Conservation laws are classified by finding a basis for the vector space of equivalence classes.

For a functional of the form 
\begin{equation}\label{eq:ddvp}
\mathscr{L}[u]=\sum_\mbn\int L(\mbx,\mbn,[\mbu])\operatorname{d}\!\mbx,
\end{equation}
where $L(\mbx,\mbn,[\mbu])$ is the differential-difference Lagrangian,
 the Euler--Lagrange equations obtained by variational calculus are $\mbE_{u^\alpha}(L)=0$. A generator $\mbv=Q^{\alpha}(\mbx,\mbn,[\mbu])\p/\p {u^{\alpha}}$ is a \textit{variational symmetry} if there exist functions $P_1^i(\mbx,\mbn,[\mbu])$ and $P_2^i(\mbx,\mbn,[\mbu])$ such that
\begin{equation}\label{eq:varsym}
\pr \mbv(L)=D_iP^i_1+\diffi P^i_2\,.
\end{equation}
From Theorem \ref{thm:ELker}, this condition amounts to
\begin{equation}\label{eq:varsym1}
	\mbE_{u^\alpha}(\pr \mbv(L))=0,\qquad \alpha=1,\dots,q,
\end{equation}
which can be used to calculate the variational symmetries whose characteristic $\mbQ$ has a given dependence on $(\mbx,\mbn,[\mbu])$.  

\begin{thm}[The differential-difference version of Noether's Theorem]
The generator $\mbv$ with characteristic $\mbQ(\mbx,\mbn,[\mbu])$ is a variational symmetry generator for the functional \eqref{eq:ddvp} if and only if $Q^\alpha\mbE_{u^\alpha}(L)$ is a conservation law for the Euler--Lagrange equations.
\end{thm}
\begin{proof}
The key identity is 
\begin{align}\label{eq:prq}
\pr \mbv(L)-Q^{\alpha}\mbE_{u^\alpha}(L)=&\ (\es_\mbK\! D_\mbJ Q^\alpha)\frac{\p L}{\p \uaJK} - Q^\alpha(-D)_\mbJ\es_{-\mbK} \left(\frac{\p L}{\p \uaJK}\right).
\end{align}
By the definition of the adjoint, the right-hand side of \eqref{eq:prq}
is a divergence.
From this identity and \eqref{eq:varsym}, if $\mbv$ is a variational symmetry generator then $Q^{\alpha}\mbE_{u^\alpha}(L)$ is a divergence that is zero on all solutions of the Euler--Lagrange equations.

Conversely, if $Q^{\alpha}\mbE_{u^\alpha}(L)$ is a conservation law, it is a divergence and hence so is
\[
\mathbf{w}(L):=(\es_\mbK\! D_\mbJ Q^\alpha)\frac{\p L}{\p \uaJK}\,.
\]
The vector field $\mathbf{w}$ is the prolongation of $\mbv=Q^\alpha \p/\p u^\alpha$, and therefore $\mbv$ is a variational symmetry generator, whose characteristic is $\mbQ$.
\end{proof}

\begin{exm}[Modified Volterra equation]\label{exm:mve}
The modified Volterra equation (see \cite{Yamilov2006})
\begin{equation}
u'=u^2(u_1-u_{-1}),
\end{equation} 
can be written as an Euler--Lagrange equation (see \cite{Peng2017}) by writing $u=1/(v_1-v_{-1})$. The Lagrangian 
\begin{equation}\label{eq:MVoltlag}
L=v_{-1}v'-\ln(v_2-v)
\end{equation}
yields the Euler--Lagrange equation
\begin{equation}\label{eq:mveel}
\mbE_v(L):= v_1'-v_{-1}'+\frac{1}{v_2-v}-\frac{1}{v-v_{-2}}=0.
\end{equation}
From \eqref{eq:varsym1}, the most general characteristic of variational Lie point symmetries for \eqref{eq:MVoltlag} is
\begin{equation}\label{eq:vecq}
	Q=\gtt^1(x)+(-1)^n\gtt^2(x),
\end{equation}
where $\gtt^1(x)$ and $\gtt^2(x)$ are arbitrary functions.
In particular, when $\gtt^1$ and $\gtt^2$ are constant, Noether's Theorem yields two conservation laws of the modified Volterra equation:
\begin{equation*}
\begin{aligned}
\mbE_v(L)=& D_x\left(\frac{1}{u}\right)+\diffn (u+u_{-1}),\\
(-1)^n\mbE_v(L)=&D_x\left(\frac{(-1)^{n}}{u}\right)+\diffn \left((-1)^{n+1}(u-u_{-1})\right).
\end{aligned}
\end{equation*}
The remaining symmetries give nonlocal conservation laws of the modified Volterra equation. For instance, $Q=x$ gives
\[
x\mbE_v(L)=D_x\!\left(\frac{x}{u}\right)+\diffn\!\left(x(u+u_{-1})-\diffn^{-1}\!\left(\frac{1}{u}\right)\right);
\]
this involves the antidifference operator $\diffn^{-1}$, as do the other nonlocal conservation laws.
\end{exm}

\section{Noether's Second Theorem}
\label{sec:NST}
Some differential-difference variational problems \eqref{eq:ddvp} have variational gauge symmetries, whose characteristics are linear homogeneous in a set of arbitrary functions of all independent variables. The differential-difference analogue of Noether's Second Theorem applies to such symmetries.

To state and prove the theorem, some definitions are needed. A \textit{differential-difference operator} on $P(J^{\infty}(\mcT_\mbn))$ is an operator of the form
\[
\mcD=f^{\mbJ;\mbK}(\mbx,\mbn,[\mbu])D_\mbJ\esK\,,
\]
for some coefficient functions $f^{\mbJ;\mbK}$. A linear \textit{differential-difference relation} (or syzygy) between \ddes $\mcA_l=0$ is an identity of the following form (for some differential-difference operators $\mcD^l$):
\[
\mcD^l \mcA_l\equiv 0,
\]
where the coefficient functions are independent of $[\mcA]$. Without loss of generality, we restrict attention from here on to characteristics $\mbQ$ that are independent of $\mcA$, so that $\mbQ|_{[\mcA=0]}=\mbQ$.

\begin{thm}[Noether's Second Theorem for \ddese] \label{ddenst}
The functional \eqref{eq:ddvp} admits a symmetry generator whose characteristic, $\mbQ(\mbx,\mbn,[\mbu;\gttb])$, is linear homogeneous in $R$ independent arbitrary functions,
\begin{equation}\label{eq:gtt}
\gttb=\left(\gtt^1(\mbx,\mbn),\dots,\gtt^R(\mbx,\mbn)\right),	
\end{equation}
if and only if there are $R$ independent linear differential-difference relations between the Euler--Lagrange equations:
\begin{equation}\label{eq:ddr}
\mcD_r^{\alpha}\mbE_{u^\alpha}(L)\equiv 0,\quad r=1,2,\ldots,R.
\end{equation}
\end{thm}
\begin{proof}
	
The proof is a simplified version of that used for the differential and difference cases in \cite{Hydon2011}. It is convenient to define an equivalence relation $\sim$ on the set of functions on $P(J^{\infty}(\mcT_\mbn))$ as follows:
\[
F_1\sim F_2\,\ \Longleftrightarrow\,\ F_1-F_2\ \text{is a divergence}.
\]
For instance, the condition for $\mbv$ to be a variational symmetry generator is $\pr\mbv(L)\sim 0.$

First, suppose that the Euler--Lagrange equations are subject to $R$ independent linear differential-difference relations \eqref{eq:ddr} and let $\gtt$ be an $R$-tuple of independent arbitrary functions \eqref{eq:gtt}. Then
\[
0\equiv\gtt^r\mcD_r^{\alpha}\mbE_{u^\alpha}(L)\sim \left\{(\mcD_r^{\alpha})^{\dagger}\gtt^r\right\} \mbE_{u^\alpha}(L)\sim \left\{D_{\mbJ}\es_{\mbK}(\mcD_r^{\alpha})^{\dagger}\gtt^r\right\}\frac{\partial L}{\partial u_{\mbJ;\mbK}^{\alpha}}=\pr\left((\mcD_r^{\alpha})^{\dagger}\gtt^r\frac{\p}{\p u^\alpha}\right)(L)\,,
\]
so each $\gtt$ yields a variational symmetry whose characteristic has components $Q^\alpha =(\mcD_r^{\alpha})^{\dagger}\gtt^r$.

Conversely, suppose that there exists a variational symmetry whose characteristic $\mbQ(\mbx,\mbn,[\mbu;\gttb])$ is linear homogeneous in $R$ independent arbitrary functions \eqref{eq:gtt}. By Noether's Theorem, $$Q^\alpha\mbE_{u^\alpha}(L)\sim 0.$$
Now treat $[\mbu]$ as subsidiary variables and vary each $\gtt^r$ in turn. Using Theorem \ref{thm:ELker} with $\gtt$ replacing $\mbu$, one obtains
\begin{equation}
\mbE_{\gtt^r}\!\left\{Q^{\alpha}(\mbx,\mbn,[\mbu;\gttb])\mbE_{u^\alpha}(L)\right\}\equiv 0,\qquad r=1,2,\ldots, R.
\end{equation}
These are the differential-difference relations among Euler--Lagrange equations; they amount to
\begin{equation}\label{eq:N2rel}
(-D)_{\mbJ}\es_{-\mbK}\left\{\frac{\partial Q^{\alpha}(\mbx,\mbn,[\mbu;\gttb])}{\partial \gtt^r_{\mbJ;\mbK}}\,\mbE_{u^\alpha}(L)\right\}\equiv 0,\qquad r=1,2,\ldots, R.
\end{equation}
As the functions $\gtt^r$ are independent, so are the relations \eqref{eq:N2rel} (which are independent of $\gtt$).
\end{proof}

\begin{cor}
Suppose that the functional \eqref{eq:ddvp} admits a symmetry generator whose characteristic, $\mbQ(\mbx,\mbn,[\mbu;\gttb])$, is linear homogeneous in $R$ independent arbitrary functions \eqref{eq:gtt}. Then the conservation laws $\mcC(\gttb)=Q^\alpha\mbE_{u^\alpha}(L)$ given by Noether's Theorem are trivial for all $\gtt$.	
\end{cor}

\begin{proof}
	As $\mbQ$ is linear homogeneous in $\gttb$ and Noether's Second Theorem gives the differential-difference relations \eqref{eq:N2rel},
	\begin{align*}
		\mcC(\gttb)&=\gtt^r_{\mbJ;\mbK}\left\{\frac{\partial Q^{\alpha}(\mbx,\mbn,[\mbu;\gttb])}{\partial \gtt^r_{\mbJ;\mbK}}\,\mbE_{u^\alpha}(L)\right\}\\
		&=\gtt^r_{\mbJ;\mbK}\left\{\frac{\partial Q^{\alpha}(\mbx,\mbn,[\mbu;\gttb])}{\partial \gtt^r_{\mbJ;\mbK}}\,\mbE_{u^\alpha}(L)\right\}-\gtt^r(-D)_{\mbJ}\es_{-\mbK}\left\{\frac{\partial Q^{\alpha}(\mbx,\mbn,[\mbu;\gttb])}{\partial \gtt^r_{\mbJ;\mbK}}\,\mbE_{u^\alpha}(L)\right\}.
	\end{align*}
This is a divergence whose components are zero when $[\mbE_{u^\alpha}(L)=0],\ \alpha=1,\dots q$.
\end{proof}

\begin{exm}

The interaction of a scalar particle of mass $m$ and charge $e$ with an electromagnetic field is described by the Euler--Lagrange equations for the Lagrangian
\begin{equation}\label{eq:nstex1la}
	L=\frac{1}{4}F_{\mu\nu}F^{\mu\nu}+(\nabla_{\mu}\psi)(\nabla^{\mu}\psi)^{\ast}+m^2\psi\psi^{\ast},
\end{equation}
where $\nabla_{\mu}=D_{\mu}+\operatorname{i}eA_{\mu}$. Here flat space-time coordinates $(x^0,x^1,x^2,x^3)$ are used ($x^0$ being time), and indices are raised or lowered by the Minkowski metric $\eta=\operatorname{diag}\{-1,1,1,1\}$. The dependent variables are the complex-valued scalar wave function $\psi$, its conjugate $\psi^{\ast}$, and the real-valued electromagnetic $4$-potential with components $A^{\mu}$, from which one obtains the antisymmetric tensor
\begin{equation}\label{eq:field}
F_{\mu\nu}=D_{\mu}A_{\nu}-D_{\nu}A_{\mu}.
\end{equation}
Adapting a finite difference method due to Christiansen \& Halvorsen \cite{ChrisHal}, one can semi-discretize the Lagrangian \eqref{eq:nstex1la} in the spatial variables $x^j,\ j=1,2,3$, leaving $t=x^0$ as the only continuous variable. (This facilitates a Hamiltonian formulation of the \ddes and a method of lines numerical solution.) Set each $x^j$ to be $h^jn^j$, where $n^j\in\mathbb{Z}$, and replace each $D_j$ in \eqref{eq:field} by the scaled difference operator $\dbj=\diffj/h^j$, whose adjoint is $\dbj^{\dagger}=-\es_j^{-1}\dbj$. The operators $\nabla_\mu$ are replaced by
\[
\overline{\nabla}_0=D_t+\operatorname{i}eA_{0},\qquad \overline{\nabla}_j=\dbj+(h^j)^{-1}\left\{1-\exp(-\operatorname{i}eh^jA_j)\right\}\id.
\]
Then the Euler--Lagrange equations are
\begin{equation*}
\begin{aligned}
&0=\mbE_{\psi}(L)=(\overline{\nabla}_{\mu})^{\dagger}(\overline{\nabla}^{\mu}\psi)^{\ast}+m^2\psi^{\ast},\\
&0=\mbE_{\psi^{\ast}}(L)=(\overline{\nabla}_{\mu}^{\,\ast})^{\dagger}(\overline{\nabla}^{\mu}\psi)+m^2\psi,\\
&0=\mbE_{A^{0}}(L)=\operatorname{Re}\left\{2\operatorname{i}e\psi\left(\overline{\nabla}_{0}\psi\right)^{\ast}\right\}+\sum_{j=1}^3\dbj^{\,\dagger}F^{0j},\\
&0=\mbE_{A^{j}}(L)=\operatorname{Re}\left\{2\operatorname{i}e\exp(-\operatorname{i}eh^jA_j)\,\psi\left(\overline{\nabla}_j\psi\right)^{\ast}\right\}-D_tF^{0j}-\sum_{k=1}^3\overline{\mathrm{D}}_k^{\,\dagger}F^{kj}.
\end{aligned}
\end{equation*}
There is a family of characteristics depending on an arbitrary real-valued function $\gtt(t,n^1,n^2,n^3)$:
\begin{equation*}
Q^{\psi}=-\operatorname{i}e\psi\gtt,\quad Q^{\psi^{\ast}}=\operatorname{i}e\psi^{\ast}\gtt,\quad Q^{A^0}=-D_t\gtt,\quad Q^{A^{j}}=\dbj\gtt.
\end{equation*}
By Theorem \ref{ddenst}, the differential-difference identity among the Euler--Lagrange equations is
\begin{equation*}
-\operatorname{i}e\psi\mbE_{\psi}(L)+\operatorname{i}e\psi^{\ast}\mbE_{\psi^{\ast}}(L)+D_t \left(\mbE_{A^0}(L)\right)+\sum_{j=1}^3\dbj^{\,\dagger}\left(\mbE_{A^{j}}(L)\right)=0.
\end{equation*}
\end{exm}

\section{An intermediate Noether-type theorem}
\label{sec:NIT}
Noether's two theorems deal with two extremes. In Noether's Theorem, a variational symmetry generator is associated with a conservation law. In Noether's Second Theorem, a family of generators depending on entirely arbitrary functions of all $p+m$ independent variables is associated with differential-difference relations between the Euler--Lagrange equations. However, many systems of interest have variational symmetries that depend on arbitrary functions (of independent variables) that have fewer than $p+m$ arguments. For instance, the variational Lie point symmetries of the Volterra equation discussed earlier depend on two arbitrary functions of $x$ only. Such systems can be treated by adapting the proof of Noether's Second Theorem to incorporate differential-difference constraints.

From here on, we consider variational symmetries whose characteristic depends on $R$ functions, $$\gtt=\left(\gtt^1(\mbx,\mbn),\gtt^2(\mbx,\mbn),\ldots,\gtt^R(\mbx,\mbn)\right),$$that are subject to a complete set of differential-difference constraints
\begin{equation}\label{eq:constr}
\msD_r^i\gtt^r(\mbx,\mbn)=0,\qquad i=1,\dots,I.
\end{equation} 
Here each $\msD^i_r$ is a linear differential-difference operator whose coefficient functions depend only on $(\mbx,\mbn)$; the set of constraints is complete if there are no integrability conditions and $\gtt$ is arbitrary apart from the constraints.

\begin{thm}\label{thm:N1pt5}
	Suppose that the functional \eqref{eq:ddvp} admits a symmetry generator whose characteristic $\mbQ(\mbx,\mbn,[\mbu;\gttb])$ is linear homogeneous in $R$ independent functions, subject to the complete set of constraints \eqref{eq:constr}. Then there exists $\mblam=(\lambda_1(\mbx,\mbn,[\mbu]),\dots,\lambda_I(\mbx,\mbn,[\mbu]))$ such that
	\begin{equation}\label{eq:deter}
	(-D)_{\mbJ}\es_{-\mbK}\left\{\frac{\partial Q^{\alpha}(\mbx,\mbn,[\mbu;\gttb])}{\partial \gtt^r_{\mbJ;\mbK}}\,\mbE_{u^\alpha}(L)\right\}+\left(\msD^i_r\right)^{\!\dagger}\lambda_i=0,\qquad r=1,\dots,R.
	\end{equation}
	For any solution $\mblam$ of \eqref{eq:deter}, there is a corresponding family of conservation laws,
	\begin{equation}\label{eq:clcon}
	\mcC_{\mblam}(\gtt):=\lambda_i\msD^i_r\gtt^r-\gtt^r\big(\msD^i_r\big)^{\!\dagger}\lambda_i\,.
	\end{equation}
	For each set of linear differential-difference operators $\mcD^r$ such that $\mcD^r(\msD^i_r)^{\!\dagger}\lambda_i\equiv 0$, there is a corresponding differential-difference relation,
	\begin{equation}\label{eq:syzelim}
	\mcD^r(-D)_{\mbJ}\es_{-\mbK}\left\{\frac{\partial Q^{\alpha}(\mbx,\mbn,[\mbu;\gttb])}{\partial \gtt^r_{\mbJ;\mbK}}\,\mbE_{u^\alpha}(L)\right\}\equiv 0.
	\end{equation}	
\end{thm}

\begin{proof}
By Noether's Theorem, $Q^\alpha\mbE_{u^\alpha}(L)$ is a divergence for each $\gtt$ satisfying the constraints, so there exists an $I$-tuple of Lagrange multipliers $\mblam$ such that
\[
\mbE_{\gtt^r}\!\left\{Q^\alpha\mbE_{u^\alpha}(L)+ \lambda_i\msD_s^i\gtt^s\right\}=0;
\]
this amounts to \eqref{eq:deter}. By definition, $\mcC_{\mblam}(\gtt)$ is a divergence. On solutions of \eqref{eq:constr} and \eqref{eq:deter},
\[
\mcC_{\mblam}(\gtt)=g^r(-D)_{\mbJ}\es_{-\mbK}\left\{\frac{\partial Q^{\alpha}(\mbx,\mbn,[\mbu;\gttb])}{\partial \gtt^r_{\mbJ;\mbK}}\,\mbE_{u^\alpha}(L)\right\},
\]
which is zero on solutions of the Euler--Lagrange equations. So $\mcC_{\mblam}(\gtt)$ is a conservation law. Differential-difference relations \eqref{eq:syzelim} are obtained only if all $\lambda_i$ can be eliminated from the system \eqref{eq:deter}. It suffices to find a generating set of linear relations, from which all others can be obtained by applying linear differential-difference operators.
\end{proof}

\begin{exm} The characteristic of the variational Lie point symmetries of the modified Volterra equation (see Example \ref{exm:mve}) can be written as $Q=\gtt(x,n)$, subject to the constraint  $\gtt_1-\gtt_{-1}=0$. The determining equation \eqref{eq:deter} for the Lagrange multiplier $\lambda$ is
	\begin{equation*}
		\mbE_\gtt\!\left\{\gtt\mbE_u(L)+\lambda(\gtt_1-\gtt_{-1})\right\}:=v_1'-v_{-1}'+\frac{1}{v_2-v}-\frac{1}{v-v_{-2}}+\lambda_{-1}-\lambda_1=0.
	\end{equation*}
	An obvious solution is
	\begin{equation*}
		\lambda=v'+\frac{1}{v_1-v_{-1}}\,,
	\end{equation*}
	which leads to the conservation law
	\begin{equation*}
		\mcC_{\mblam}=\diffn\! \left\{\gtt\left(v_{-1}'+\frac{1}{v-v_{-2}}\right)+\gtt_{-1}\left(v'+\frac{1}{v_1-v_{-1}}\right)\right\}.
	\end{equation*}
As $\mcC_{\mblam}=Q\mbE_u(L)$ for all $\gtt$ satisfying the constraint, this is precisely the family of conservation laws given by Noether's Theorem. 
\end{exm}

\begin{exm}
The Euler--Lagrange equations for the Lagrangian
\[
L=(u_1-v-\tfrac{1}{2}w')w'+v(u_1-u)
\]
constitute the following linear system for $u(x,n),\, v(x,n)$ and $w(x,n)$:
\[ 
\mbE_u(L):=w_{-1}'+v_{-1}-v=0,\qquad \mbE_v(L):=-w'+u_1-u=0,\qquad \mbE_w(L):=-u_1'+v'+w''=0.
\]
This system has variational symmetries generated by
\[
\mbv=\gtt^1(x,n)\frac{\p}{\p u}+\gtt^2(x,n)\frac{\p}{\p v}+\gtt^3(x,n)\frac{\p}{\p w}\,,
\]
subject to the complete set of constraints
\[
D_x\gtt^3-\diffn\gtt^1=0,\qquad D_x(\gtt^2-\gtt^1)=0,\qquad \diffn(\gtt^2-\gtt^1)=0.
\]
Then system \eqref{eq:deter} determining the Lagrange multipliers for these constraints is
\begin{align*}
	0&=\mbE_u(L)+\es^{-1}\!\diffn\lambda_1+D_x\lambda_2+\es^{-1}\!\diffn\lambda_3\,,\\
	0&=\mbE_v(L)-D_x\lambda_2-\es^{-1}\!\diffn\lambda_3\,,\\
	0&=\mbE_w(L)-D_x\lambda_1.
\end{align*}
Eliminating all $\lambda_i$, we obtain a linear differential relation between the Euler--Lagrange equations that generates all other such relations:
\[
D_x(\mbE_u(L))+D_x(\mbE_v(L))+\es^{-1}\!\diffn(\mbE_w(L))\equiv 0.
\]
The determining system \eqref{eq:deter} has a solution,
\[
\lambda_1=-u_1+v+w',\qquad \lambda_2=-w,\qquad \lambda_3=u_1,
\]
which yields the family of conservation laws
\[
\mcC_{\mblam}(\gtt)=D_x\!\left\{(\gtt^1-\gtt^2)w+\gtt^3(-u_1+v+w')\right\}+\diffn\!\left\{-\gtt^1(v_{-1}+w_{-1}')+\gtt^2 u\right\}.
\]
The constraints can be solved explicitly in terms of an arbitrary function $f(x,n)$ and an arbitrary constant $c$:
\[
\gtt^1=f',\qquad \gtt^2=f'+c,\qquad \gtt^3=\diffn f.
\]
Then the family of conservation laws $\mcC_{\mblam}(\gtt)$ can be rewritten as
\[
\mcC_{\mblam}(\gtt)=D_x\!\left\{-cw+f(\mbE_u(L)+\mbE_v(L))\right\}+\diffn\!\left\{cu+f\es^{-1}\mbE_w(L)\right\}.
\]
Up to a trivial conservation law, every conservation law in this family is a multiple of the non-trivial conservation law $\mbE_v(L)=D_x(-w)+\diffn(u)$. (The family includes the obvious conservation laws $\mbE_u(L)$ and $\mbE_w(L)$.)
So the intermediate Noether-type theorem \ref{thm:N1pt5} reveals that this simple example has both one generating differential-difference relation and a one-dimensional equivalence class of conservation laws. Neither of these facts is immediately obvious from the Euler--Lagrange equations.
\end{exm}

\section{Conclusions}

By establishing the prolongation structure that must be preserved by a point transformation, we have shown that for Lie point symmetries of a generic \ddee, the transformed continuous independent variables cannot depend on the discrete independent variables. (By contrast, the transformed dependent variables can depend on all dependent and independent variables.) There is an exceptional class, partitioned \ddese, that allow such dependence, provided that it is compatibly periodic. These discoveries resolve the problem that differentiation and shifting do not commute if a mapping is not structure-preserving.

It can be shown that the approach taken by Levi \textit{et al}. in \cite{Levi2010} amounts to fixing an arbitrary point $\mbn$, then using the transformed derivative for that $\mbn$ to calculate prolongations of the variables $\mbu$ over all other discrete points, with a correction factor that accounts for the discrepancy between the discrete points. This approach has the advantage of allowing the transformed $\mbx$ to depend on $\mbn$ (at least, in principle). However, it turns out that the symmetries in all classes of equations investigated in \cite{Levi2010} do not exhibit such dependence, and the circumstances in which this dependence can occur are unknown. Moreover, the transformed derivatives in \cite{Levi2010} depend on which $\mbn$ is chosen, so that, unlike the current paper, there is no well-defined total prolongation space (which is foundational in both the differential and difference cases).

By factoring out trivial symmetries, we have proved differential-difference analogues of Noether's two theorems on variational symmetries, and have established an intermediate theorem that applies when the characteristic depends on functions that are subject to linear differential-difference constraints.

\section*{ Acknowledgments} 

 The authors are grateful to Centre International de Rencontres Math\'{e}matiques in Marseille for support and hospitality during the conference {\it `Symmetry and Computation'}, when work of this paper was initiated. The authors would like to thank the Isaac Newton Institute for Mathematical Sciences for support and hospitality during the programme {\it `Geometry, Compatibility and Structure Preservation in Computational Differential Equations'} when some work on this paper was undertaken. This work was partially supported by JSPS KAKENHI Grant Number JP20K14365, JST-CREST Grant Number JPMJCR1914, Keio Gijuku Academic Development Funds, Keio Gijuku Fukuzawa Memorial Fund, Waseda University Grant Program for Promotion of International Joint Research, and EPSRC Grant Number EP/R014604/1. We thank the referees for their constructive comments and suggestions.

\end{document}